\documentclass[review]{elsarticle}

\usepackage{amssymb}
\usepackage{amsmath}
\newtheorem{theorem}{Theorem}
\newtheorem{lemma}{Lemma}
\newtheorem{definition}{Definition}
\newtheorem{proof}{Proof}
\newtheorem{remark}{Remark}
\usepackage{lineno,hyperref}
\modulolinenumbers[5]

\journal{Journal of \LaTeX\ Templates}

\begin{document}

\begin{frontmatter}

\title{Query complexity of unitary operation discrimination}
\author{Xiaowei Huang}
\author{Lvzhou Li\corref{mycorrespondingauthor}}
\address{Institute of Quantum Computing and Computer Theory, School of Computer and Engineering, Sun Yat-sen University, Guangzhou 510006, China}
\cortext[mycorrespondingauthor]{Corresponding author}
\ead{lilvzh@mail.sysu.edu.cn}





\begin{abstract}
Discrimination of unitary operations is fundamental in quantum computation
and information. A lot of quantum algorithms including the well-known 
Deutsch-Jozsa algorithm, Simon's algorithm, and Grover's algorithm can essentially
be regarded as discriminating among individual, or sets of unitary operations
(oracle operators). The problem of discriminating between two unitary operations
$U$ and $V$ can be described as: Given $X\in\{U, V\}$, determine which one 
$X$ is. If $X$ is given with multiple copies, then one can design an adaptive
procedure that takes multiple queries to $X$ to output the identification result
of $X$. In this paper, we consider the problem: How many queries are required 
for achieving a desired failure probability $\epsilon$ of discrimination between
$U$ and $V$. We prove in a uniform framework: (i) if $U$ and $V$ are discriminated 
with bound error $\epsilon$ , then the number of queries $T$ must satisfy 
$T\geq \left\lceil\frac{2\sqrt{1-4\epsilon(1-\epsilon)}}{\Theta (U^\dagger V)}\right\rceil$,
and (ii) if they are discriminated with one-sided error $\epsilon$, then there is
$T\geq \left\lceil\frac{2\sqrt{1-\epsilon^2}}{\Theta (U^\dagger V)}\right\rceil$, 
where $\lceil k\rceil$ denotes the minimum integer not less than $k$ and 
$\Theta(W)$ denotes the length of the smallest arc containing all the eigenvalues of 
$W$ on the unit circle.
\end{abstract}

\begin{keyword}
Quantum Computing \sep Unitary Operation Discrimination \sep  Quantum Query Complexity
\MSC[2020] 68Q12 \sep 81P45 
\end{keyword}

\end{frontmatter}


\section{Introduction}
Discrimination of  unitary operations  plays a  fundamental role in the field of quantum information science, as many quantum information processing tasks  eventually involve discriminating among unitary operations.  Specially, a quantum algorithm  can be essentially regarded as  a procedure to discriminate among individual, or sets of  unitary operations.  For example, the well-known Deutsch-Jozsa algorithm \cite{deutsch1992rapid} is  to distinguish  two set of unitary operations  $S_B$ and $S_C$ as follows:
\begin{align*}S_B&=\{O_x:x\in\{0,1\}^n, |x|=\frac{n}{2}\},\\
      S_C&=\{O_x:x\in\{0,1\}^n, |x|=0~ \text{or}~ 1\}
\end{align*}
where $S_B$ stands for the balanced functions, $S_C$ stands for the constant functions, and $O_x=\sum_i(-1)^{x_i}|x_i\rangle\langle  x_i|$ is the  oracle operator.
Discrimination among sets of unitary operations (oracle operators) standing for functions with different periodicities is central to  Simon's algorithm \cite{simon1994proceedings}
and Shor's algorithm \cite{shor199435th}. Actually, in the quantum query complexity model, any quantum algorithm for computing a Boolean function is  a procedure of discriminating between two sets of oracle operators.

Considering the significance of unitaries discrimination, it has been studied in depth (for a partial list \cite{acin2001statistical,d2001using,PhysRevLett.98.129901,PhysRevLett.99.170401,PhysRevLett.100.020503,PhysRevA.77.032337, PhysRevA.96.022303,chiribella2013identification, bae2015discrimination,cao2016minimal,PhysRevA.96.052327,kawachi2019quantum,cao2016local,cao2016determination}), and the discrimination problem has also been discussed for measurements \cite{PhysRevLett.96.200401}, Pauli channels \cite{PhysRevA.72.052302},   oracle operators \cite{chefles2007unambiguous}, and the general quantum operations \cite{PhysRevA.71.062340, li2008optimal, PhysRevA.73.042301,zbMATH05507378,PhysRevLett.103.210501}.

Discrimination of  unitary operations is generally transformed to discrimination of quantum states by preparing a probe state and then identifying the output states generated by different unitary operations.  Two unitary operations $U$ and $V$ are said to be
perfectly distinguishable (with  a single query), if there
exists a state $|\psi\rangle$ such that $U|\psi\rangle\perp V|\psi\rangle$.   It has been  shown that  $U$ and
$V$ are perfectly distinguishable
if, and only if $\Theta(U^\dagger V)\geq\pi$, where $\Theta(W)$
denotes the length of the smallest arc containing all the
eigenvalues of $W$ on the unit circle \cite{acin2001statistical,d2001using}. The situation changes dramatically when  multiple queries  are allowed, since any two different unitary operations are perfectly distinguishable in this case. Specifically,
it was shown that for any two different unitary
operations $U$ and $V$, there exist a finite number $N$ and a
suitable state $|\varphi\rangle$ such that $U^{\otimes N}|\varphi\rangle\perp V^{\otimes N}|\varphi\rangle$  \cite{acin2001statistical,d2001using}.
Such a discriminating scheme is intuitively called a {\it parallel scheme}. Note that in the parallel scheme, an $N$-partite
entangled state as an input is required  and plays a crucial role. Then, the
result was further refined in Ref. \cite{PhysRevLett.98.129901} by showing that the
entangled input state is not necessary for perfect discrimination of unitary operations. Specifically, Ref. \cite{PhysRevLett.98.129901} showed
that for any two different unitary operations $U$ and $V$, there
exist an input state $|\varphi\rangle$ and auxiliary unitary operations
$W_1,\dots,W_N$ such that $UW_NU\dots W_1U|\varphi\rangle\perp VW_NV\dots W_1V|\varphi\rangle$.  Such
a discriminating scheme is generally called a {\it sequential scheme}.

Note that in these researches mentioned above, it was assumed by default that the  unitary operations to be discriminated are under the
complete control of a single party who can  perform any physically
allowed operations to achieve an optimal discrimination. A more complicated case is that the
 unitary operations to be discriminated   are shared by
several spatially separated parties. Then, in this case a reasonable constraint
on the discrimination is that each party can only make local
operations and classical communication (LOCC). Despite this constraint, it has been shown that  any two bipartite unitary
operations can be perfectly discriminated by LOCC, when   multiple  queries to the unitary operations are allowed \cite{PhysRevLett.99.170401,PhysRevLett.100.020503,PhysRevA.77.032337, PhysRevA.96.022303,PhysRevA.96.052327,bae2015discrimination}.

All the above mentioned works focus on the perfect discrimination of unitary operations. If a failure probability can be tolerated, then the problem can be discussed in two more general cases: minimum-error  discrimination where a non-zero probability
of the identification result being erroneous is tolerated, and  unambiguous  discrimination where the  the identification result is always correct but we are not
always guaranteed a conclusive result. Specifically,  in the minimum-error  discrimination case, two unitary operations $U$ and $V$ are said to be discriminated with bounded error $\epsilon$, if each of them can be identified correctly with probability more than $1-\epsilon$; in the unambiguous  discrimination case,  $U$ and $V$ are said to be   discriminated with one-sided error $\epsilon$, if for each one, either it is identified correctly, or the discrimination procedure output  with probability less than $\epsilon$ the inconclusive result "I do not know".

Now suppose a unitary operation $X$ is selected from $\{U_1, U_2\}$ and can be accessed multiple times. Then a general procedure to identify $X$ (or discriminate between $U_1$ and $U_2$) is depicted in Fig. \ref{Fig_2}, where one adaptively performs a series of $X$ and some other auxiliary operations to an initial  state and then performs a measurement with  a result as the identifier of $X$. The problem considered in this paper is: How many queries are required for achieving a  desired failure probability $\epsilon$ of discrimination between $U_1$ and $U_2$. We prove in a uniform framework that  (i) if $U_1$ and $U_2$ are discriminated with bounded error $\epsilon$ (resp. with one-sided error $\epsilon$), then they need to be accessed by at least  $\left\lceil\frac{2\sqrt{1-4\epsilon(1-\epsilon)}}{\Theta (U_1^\dagger U_2)}\right\rceil$ queries (resp. $\left\lceil\frac{2\sqrt{1-\epsilon^2}}{\Theta (U_1^\dagger U_2)}\right\rceil$ queries).

The rest of this paper is organized as follows. Section \ref{sec:1} present some preliminary knowledge especially for fidelity of two unitaries. The main result   is presented in Section \ref{sec:2}. A conclusion is made in Section \ref{Con}.

\section{Preliminaries} \label{sec:1} One can refer to an excellent textbook \cite{nielsen2010quantum} for the details of quantum computation and quantum information. According to quantum mechanics, the evolution of a closed quantum system is described by a unitary operation (transformation). An operation $U$ is said to be unitary if $UU^\dagger=U^\dagger U=I$ where $U^\dagger$ denotes the adjoint of $U$. 
The fidelity of two unitary operations $U_1$ and $U_2$ is \footnote{In mathematical form, the fidelity should be defined as $F(U,V) \equiv\min_{|\Psi\rangle} |\langle\Phi| U^\dagger V\otimes I|\Phi\rangle|$, but  it is easy to see that the minimum value is achieved in Eq. (\ref{fd}).}
 \begin{align}F(U_1, U_2)\equiv\min_{|\psi\rangle}|\langle\psi| U_1^\dagger U_2|\psi\rangle|. \label{fd} \end{align} 
 If we let
 $U_1^\dagger U_2=\sum_{j=1}^{k}e^{i\theta_j}|j\rangle \langle j|$,
and  $|\psi\rangle=\sum_j \lambda_j|j\rangle$ with $\sum_{j}|\lambda_j|^2 = 1$, then there is
 \begin{align}
 F(U_1, U_2)
 &=\min\left\{ \left| \sum_{j}|\lambda_j|^2
 e^{i\theta_j} \right| : \sum_{j}|\lambda_j|^2 = 1 \right\}.
 \label{eq:sigmax}
 \end{align}
 Denote the  convex hull of $S=\left\{e^{i\theta_j} \right\}_j$ by   \begin{align}conv\left(S\right)=\left\{  \sum_{j}|\lambda_j|^2 e^{i\theta_j}: \sum_{j}|\lambda_j|^2 = 1 \right\}. \end{align}
 Then the fidelity can be represented by
 \begin{equation}
 F(U_1, U_2)=\min_{P\in conv\left(S\right)}||O-P||,
 \end{equation}
which states that $F(U_1, U_2)$ corresponds to the minimum distance from the origin $O$  to the convex hull $conv\left(S\right)$.

\begin{figure}
	\center \includegraphics[width=8cm]{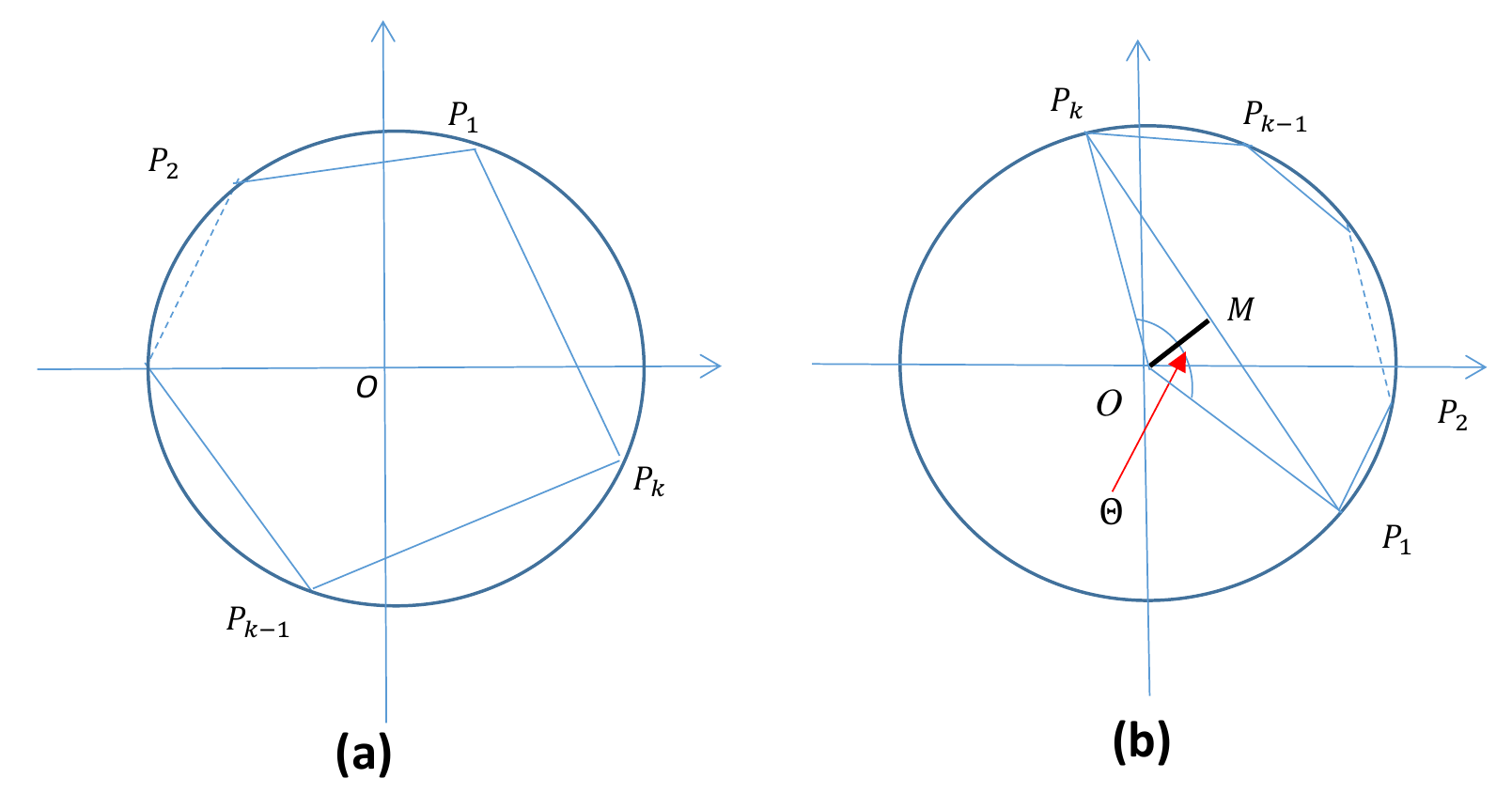}
	\caption{ The convex hull $conv(S)$. $P_j$ corresponds to  $e^{i\theta_j}$ with $j=1,\dots, k$, where without loss of generality, we assume  these points on the unit circle are $P_1, \dots, P_k$ in the counter-clockwise order.  In (a), the convex hull contains the original point $O$, which implies that $F(U_1, U_2)=0$. In (b), the convex hull does not contain $O$. Then $F(U_1, U_2)$ is equal to the distance between $O$ and $M$, where $M$ denotes the point in $conv(S)$ that has a minimum distance to the origin $O$.} \label{Fig 1}
\end{figure}

In geometry, each $e^{i\theta_j}$ stands for a point on the unit circle in the complex plane. As shown in Fig. \ref{Fig 1}, let $P_j$ denote the point $e^{i\theta_j}$ with $j=1,\dots,k$. Without loss of generality, assume the  counter-clockwise order of these points on the unit circle is $P_1, P_2, \dots, P_k$.   Denote  by $\diamond P_1 \dots P_k$ the region enclosed by the convex  polygon with endpoints $P_1,  \dots, P_k$. Then  $\diamond P_1\dots P_k$ is the convex hull $conv(S)$. By this geometry representation, we have
\begin{align}
F(U_1, U_2)=\begin{cases}
|OM| & O\not\in conv(S),\\
0 &   O\in conv(S).
\end{cases}
\end{align}

\begin{figure}
	\center \includegraphics[width=6cm]{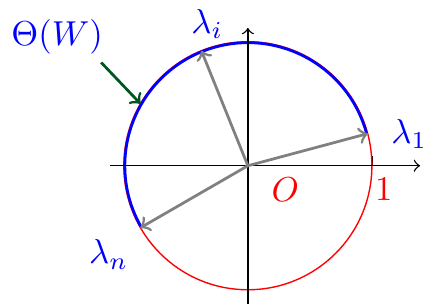}
	\caption{  The eigenvalues of $W$ on the unit circle, $\Theta(W)$ colored in blue.} \label{Fig 3}
\end{figure}

Let $\Theta(W)$
denote the length of the smallest arc containing all the
eigenvalues of unitary operation $W$ on the unit circle (depicted in Fig. \ref{Fig 3}). Then  $F(U_1, U_2)$ can be rewritten as
\begin{align}
F(U_1, U_2)
=\begin{cases}
\cos \frac{\Theta(U_1^\dagger U_2)}{2} & 0\leq \Theta(U_1^\dagger U_2)<\pi,\\
0 &   \Theta(U_1^\dagger U_2)\geq \pi.
\end{cases}\label{F}
\end{align}

The trace norm of operator $A$ is defined as
$||A||_{tr}=\text{Tr}\sqrt{A^\dagger A}$.
In this way, $\| |\varphi\rangle\langle\varphi|- |\phi\rangle\langle\phi|\|_{tr}$ denotes the trace distance between $ |\varphi\rangle$ and $|\phi\rangle$. For simplicity, we  denote it by  $\| |\varphi\rangle- |\phi\rangle\|_{tr}$ throughout this paper, and it can be verified that  \begin{equation}
\| |\varphi\rangle- |\phi\rangle\|_{tr}=2\sqrt{1-|\langle\varphi|\phi\rangle|^2}.
\end{equation}

\section{Lower bound on  query complexity}\label{sec:2}

A general procedure to discriminate two unitary operations $U_1$ and $U_2$  can be seen as a sequence of unitaries $W_T U_iW_{T-1}U_i\cdots W_1U_iW_0$, where $W_k$'s are fixed unitaries and $T$ is called the query complexity to $U_i (i=1,2)$.
As shown in Fig. 	\ref{Fig_2}, the discrimination process is as follows:
\begin{itemize}
	\item[\textbf{1.}]  Start with an initial state $|\varphi_{0}\rangle$.
	
	\item[\textbf{2.}] Perform the operators $W_0, U_i, W_1, U_i,\dots ,W_T$ in sequence, and then  obtain the following state  corresponding  to the chosen $U_i$:
 \begin{align}
|\Phi_i\rangle=W_T (U_i\otimes I)W_{T-1} (U_i\otimes I)\cdots W_0|\varphi_{0}\rangle.
\end{align}

	\item[\textbf{3.}] A $2$-outcome measurement is performed on $|\Phi_i\rangle$ with  a result  as the identifier of $U_i$.
\end{itemize}

\begin{figure}
	\center \includegraphics[width=9cm]{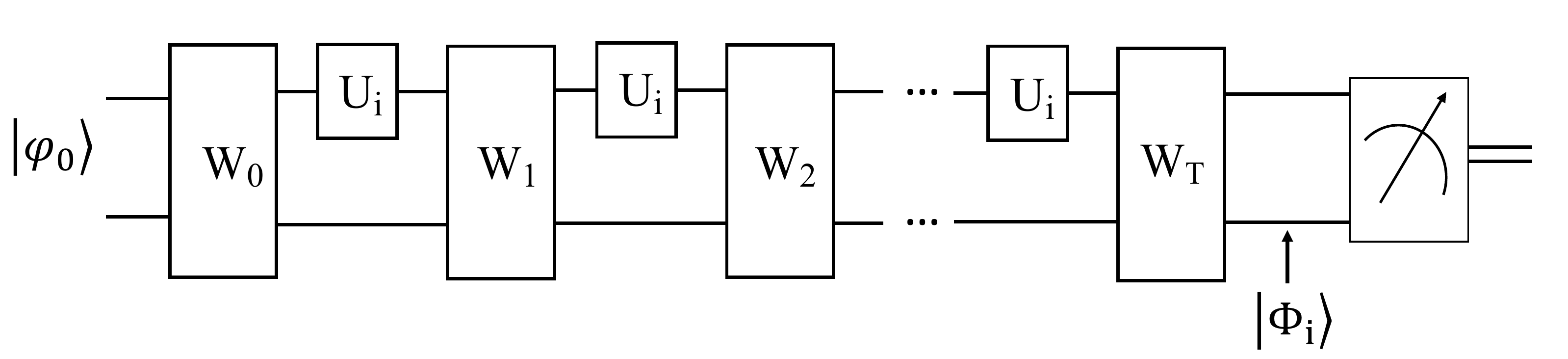}
	\caption{ A $T$-query procedure to discriminate  two unitary operations. $U_i (i=1,2)$ are the unitaries to be discriminated. $W_0, W_1, \dots ,W_T$ are some fixed unitaries. $T$ is called the query complexity.} \label{Fig_2}
\end{figure}

\begin{definition}

(i) $U_1$ and $U_2$ are said to be discriminated with bounded error $\epsilon$, if there exists a procedure  as shown in Fig.  \ref{Fig_2} and a POVM  measurement $\{\Pi_1,\Pi_2\}$ such that $\langle \Phi_i |\Pi_i|\Phi_i\rangle\geq 1-\epsilon $ for $i=1,2$.

(ii) $U_1$ and $U_2$ are said to be discriminated with one-sided error $\epsilon$, if there exists a procedure  as shown in Fig.  \ref{Fig_2} and a POVM measurement $\{\Pi_0,\Pi_1,\Pi_2\}$ such that $\langle \Phi_2 |\Pi_1|\Phi_2\rangle=\langle \Phi_1 |\Pi_2|\Phi_1\rangle=0$  and $\langle\Phi_i|\Pi_0|\Phi_i\rangle\leq \epsilon$ for $i=1,2$. \label{Df1}
\end{definition}

 The problem considered in this paper is: How many queries are required for achieving a  desired error $\epsilon$ of discrimination between  $U_1$ and $U_2$. Our main result is given in Theorem \ref{Thm1}. In order to prove that, some intermediate results (i.e., Lemma \ref{Lm1}, \ref{Lm2} and \ref{Lm3}) are required.

\begin{lemma}
(i) If $U_1$ and $U_2$ can be discriminated with bounded error $\epsilon$, then $|\langle\Phi_1|\Phi_2\rangle|\leq 2\sqrt{\epsilon(1-\epsilon)}$; (ii) If $U_1$ and $U_2$ can be discriminated with  one-sided  error $\epsilon$, then $|\langle\Phi_1|\Phi_2\rangle|\leq{\epsilon}.$ \label{Lm1}
\end{lemma}

\begin{proof}
Item (i) is essentially a  known and widely used  result. For completeness, we give a proof in the following. If we denote $P_s=\langle \Phi_1 |\Pi_1|\Phi_1\rangle +\langle \Phi_2 |\Pi_2|\Phi_2\rangle $, then we have
\begin{align}
P_s&= \langle \Phi_1 |(I- \Pi_2)|\Phi_1\rangle +\langle \Phi_2 |\Pi_2|\Phi_2\rangle \\
&=1+Tr\left(\Pi_2(|\Phi_2\rangle\langle\Phi_2|-|\Phi_1\rangle\langle\Phi_1|)\right).
\end{align}
Similarly, we have	
\begin{align}
P_s=1-Tr\left(\Pi_1(|\Phi_2\rangle\langle\Phi_2|-|\Phi_1\rangle\langle\Phi_1|)\right).
\end{align}
Thus there is
\begin{align}
P_s&=1+\frac{1}{2}Tr\left((\Pi_2-\Pi_1)|\Phi_2\rangle\langle\Phi_2|-|\Phi_1\rangle\langle\Phi_1|\right)\\
&\leq 1+\frac{1}{2}  \|\Phi_2\rangle\langle\Phi_2|-|\Phi_1\rangle\langle\Phi_1|\|_{tr}\\
&=1+ \sqrt{1-|\langle\Phi_2|\Phi_1\rangle|^2}
\end{align}
where the inequality holds by observing that $|\Phi_2\rangle\langle\Phi_2|-|\Phi_1\rangle\langle\Phi_1|=Q-S$ for two positive operators and $\|\Phi_2\rangle\langle\Phi_2|-|\Phi_1\rangle\langle\Phi_1|\|_{tr}=Tr(Q)+Tr(S)$.
On the other hand, by  item (i) of Definition \ref{Df1},  there is $P_s\geq 2(1-\epsilon)$. Thus we get $|\langle\Phi_1|\Phi_2\rangle|\leq 2\sqrt{\epsilon(1-\epsilon)}$.

It is easy to see item (ii)  by observing the following:
\begin{align*}
|\langle\Phi_1|\Phi_2\rangle|&=|\langle\Phi_1|(\Pi_0+\Pi_1+\Pi_2)|\Phi_2\rangle|\\
&\leq |\langle\Phi_1|\Pi_0 |\Phi_2\rangle|+|\langle\Phi_1|\Pi_1 |\Phi_2\rangle|+|\langle\Phi_1|\Pi_2 |\Phi_2\rangle|\\
&\leq \sqrt{\langle\Phi_1|\Pi_0 |\Phi_1\rangle} \sqrt{\langle\Phi_2|\Pi_0 |\Phi_2\rangle}\\
&\leq \sqrt{\epsilon}. \sqrt{\epsilon}=\epsilon
\end{align*}
where the second inequality follows from $|\langle\Phi_1|\Pi_1 |\Phi_2\rangle|=|\langle\Phi_1|\Pi_2 |\Phi_2\rangle|=0$ (implied in item (ii) of Definition \ref{Df1}) and the application of Cauchy-Schwarz inequality to $|\langle\Phi_1|\Pi_0 |\Phi_2\rangle|$.
\end{proof}

For a $T$-query discrimination procedure as shown in Fig. \ref{Fig_2}, denote
\begin{align}|\varphi^k_i\rangle=W_k(U_i\otimes I)W_{k-1}(U_i\otimes I)\cdots W_0|\varphi_0\rangle\end{align} for $i=1,2$, and $k=0,1,\dots, T$. Note that $|\varphi^0_i\rangle=W_0|\varphi_0\rangle$ and
\begin{align}
|\varphi^{k+1}_i\rangle= W_{k+1}(U_i\otimes I)|\varphi^k_i\rangle.
\end{align}

Denote \begin{align}
D_k=\| |\varphi_1^k\rangle- |\varphi_2^k\rangle\|_{tr}=2\sqrt{1-|\langle\varphi_1^k|\varphi_2^k\rangle|^2}.
\end{align}
Then we obtain the following crucial result.
\begin{lemma}
 $D_0=0$ and $D_{k+1}\leq D_k+2\sqrt{1-F^2(U_1, U_2)}$ for $k=0,1,\dots,T-1$. \label{Lm2}
\end{lemma}

\begin{proof} First it is easy to see that $D_0=0$. For $D_{k+1}$ we have
\begin{align*}
D_{k+1}&=\| |\varphi_1^{k+1}\rangle- |\varphi_2^{k+1}\rangle\|_{tr}\\
&=\| W_{k+1}(U_1\otimes I)|\varphi_1^{k}\rangle-W_{k+1}(U_2\otimes I)|\varphi_2^{k}\rangle\|_{tr}\\
&=\| (U_1\otimes I)|\varphi_1^{k}\rangle- (U_2\otimes I)|\varphi_2^{k}\rangle\|_{tr}\\
&=\| (U_1\otimes I)|\varphi_1^{k}\rangle- (U_1\otimes I)|\varphi_2^{k}\rangle+(U_1\otimes I)|\varphi_2^{k}\rangle-(U_2\otimes I)|\varphi_2^{k}\rangle\|_{tr}\\
&\leq \| |\varphi_1^{k}\rangle- |\varphi_2^{k}\rangle\|_{tr}+ \|(U_1\otimes I)|\varphi_2^{k}\rangle-(U_2\otimes I)|\varphi_2^{k}\rangle\|_{tr}\\
&=D_k+2\sqrt{1-|\langle\varphi_2^{k}| U_1^\dagger U_2\otimes I |\varphi_2^{k}\rangle|^2}\\
&\leq D_k+2\sqrt{1-F^2(U_1,U_2)},
\end{align*}
where the first inequality follows from the triangle inequality and the unitary invariance of trace distance. The last inequality follows from the fact that $F(U,V) \leq |\langle\Phi| U^\dagger V\otimes I|\Phi\rangle|$ for any $|\Phi\rangle$ which is implied by the definition in Eq. (\ref{fd}).
\end{proof}
\begin{lemma} Suppose $U_1, U_2$ are discriminated by using $T$ queries. Then we have:
(i) if $U_1, U_2$ are discriminated with bounded error $\epsilon$, then $D_T\geq 2\sqrt{1-4\epsilon(1-\epsilon)}$;
(ii) if  $U_1, U_2$ are discriminated with one-sided  error $\epsilon$, then $D_T\geq 2\sqrt{1-\epsilon^2}$. \label{Lm3}
\end{lemma}
\begin{proof}
\begin{align*}
D_T&=\| |\varphi_1^T\rangle- |\varphi_2^T\rangle\|_{tr}=\| |\Phi_1\rangle- |\Phi_2\rangle\|_{tr}\\
&=2\sqrt{1-|\langle\Phi_1|\Phi_2\rangle|^2}.
\end{align*}
Combining the above formula with Lemma \ref{Lm1}, we get the result.
\end{proof}

Now we are in the position to present our main result and its proof.
\begin{theorem}
	Suppose $U_1, U_2$ are discrimited by using $T$ queries. Then we have:
	
	(i) if they are discriminated with bounded error $\epsilon$, then $T\geq \left\lceil \frac{2\sqrt{1-4\epsilon(1-\epsilon)}}{\Theta (U_1^\dagger U_2)}\right\rceil$;
	
	(ii) if they are discriminated with one-sided error $\epsilon$, then $T\geq \left\lceil\frac{2\sqrt{1-\epsilon^2}}{\Theta (U_1^\dagger U_2)}\right\rceil$. \label{Thm1}
\end{theorem}
\begin{proof}
 Below we prove the result for case (i): $U_1, U_2$ are discriminated with bounded error $\epsilon$. The other case can be proved similarly.  First, by Lemma \ref{Lm2} we have
\begin{align*}
D_T&=(D_T-D_{T-1})+(D_{T-1}-D_{T-2})+(D_1-D_0)+D_0\\
&\leq 2T \sqrt{1-F^2(U_1,U_2)}.
\end{align*}
Then by Lemma  \ref{Lm3}, there is
\begin{align}
2\sqrt{1-4\epsilon(1-\epsilon)}\leq D_T \leq 2 T \sqrt{1-F^2(U_1,U_2)},
\end{align}
which leads to
\begin{align}
T\geq\sqrt{ \frac{1-4\epsilon(1-\epsilon)}{1-F^2(U_1,U_2)}}. \label{T1}
\end{align}

If $F(U_1, U_2)=0$, then $U_1$ and $U_2$ can be perfectly discriminated with one query. Thus,  we need only consider the case of $F(U_1, U_2)\neq 0$. Then by Eq. (\ref{F}), we have $F(U_1, U_2)=\cos \frac{\Theta(U_1^\dagger U_2)}{2}$, which leads to
\begin{align}
\sqrt{1-F^2(U_1,U_2)}&=\sqrt{1-\cos^2 \frac{\Theta(U_1^\dagger U_2)}{2}}\\
&=\sin \frac{\Theta(U_1^\dagger U_2)}{2}\\
&\leq \frac{\Theta(U_1^\dagger U_2)}{2}.\label{F1}
\end{align}
Combining (\ref{F1}) and (\ref{T1}), we obtain
\begin{align}T\geq \frac{2\sqrt{1-4\epsilon(1-\epsilon)}}{\Theta(U_1^\dagger U_2)}.\end{align}
\end{proof}

\begin{remark}
If $U_1, U_2$ are required to discriminated perfectly (that is, $\epsilon=0$), then both the two cases lead to $ T\geq \frac{2}{\Theta(U_1^\dagger U_2)}.$ This is consistent with the result given in \cite{PhysRevLett.98.129901} that two unitaries $U_1, U_2$ can be discriminated perfectly by making queries to the unitaries $\left\lceil\frac{\pi}{\Theta(U_1^\dagger U_2)}\right\rceil$ times.
\end{remark}

Compared with a related work \cite{kawachi2019quantum}, two main differences are as follows: (i) We have obtained the query complexity lower bounds in a uniform framework for both the bounded error case  and the one-sided error case, whereas Ref. \cite{kawachi2019quantum} considered only the former case. (ii) We have presented a lower bound for an arbitrary $\epsilon$, but the proof in \cite{kawachi2019quantum} has dependence on the specific value $\epsilon=\frac{1}{3}$.

\section {Conclusion} \label{Con}

We have considered the query complexity of discrimination between two unitary operations $U$ and $V$.  It is proved that (i) for minimum-error discrimination, at least $\left\lceil\frac{2\sqrt{1-4\epsilon(1-\epsilon)}}{\Theta (U^\dagger V)}\right\rceil$ queries are required to discriminate between  $U$ and $V$ with bounded error $\epsilon$, and (ii) for ambiguous discrimination, at least  $\left\lceil\frac{2\sqrt{1-\epsilon^2}}{\Theta (U^\dagger V)}\right\rceil$ queries are required to discriminate between  them with one-sided error $\epsilon$. The proof is presented in a uniform framework.


\bibliography{refs}

\end{document}